\documentclass{amsart}
\date{March 20, 2014}

\newtheorem{theorem}{Theorem}[section]

\newtheorem{proposition}[theorem]{Proposition}

\theoremstyle{definition}

\theoremstyle{remark}
\newtheorem{remark}[theorem]{Remark}

\numberwithin{equation}{section}

\newcommand{\rd}{{\,\rm d}}
\newcommand{\e}{{\rm e}}

\newcommand{\Ran}{\mathop{\rm Ran}}

\newcommand{\R}{{\mathbb R}}
\newcommand{\C}{{\mathbb C}}

\newcommand{\dom}{\mathcal{D}}

\newcommand\re{\mathrm{Re}}

\newcommand\I{\mathrm{i}}
\newcommand{\set}[2]{\{#1 : #2 \}}

\renewcommand\H{\mathcal{H}}

\DeclareMathOperator{\supp}{supp}
\DeclareMathOperator{\tr}{tr}

\begin{document}

\title[Dipoles in Graphene]{Dipoles in Graphene Have Infinitely Many Bound States}

\author{Jean-Claude Cuenin}
\author{Heinz Siedentop}

\address{Mathematisches Institut\\
 Ludwig-Maximilians-Universit\"at M\"unchen\\
 Theresienstra\ss e 39\\
 80333 Munich\\ Germany}
 
\email{cuenin@math.lmu.de}

\email{h.s@lmu.de}

\begin{abstract}
  We show that in graphene charge distributions with non-vanishing
  dipole moment have infinitely many bound states. The corresponding
  eigenvalues accumulate at the edges of the gap faster than any
  power.
\end{abstract}

\maketitle

\section{Introduction}

Graphene close to the Fermi surface is often described by 
two-dimensional massless Dirac operators. Strained graphene, though develops a
mass gap (Vozmediano et al~\cite{Vozmedianoetal2010}). These materials
together with an electric dipole recently attracted attention by De
Martino et al \cite{Martinoetal2014}. They predicted that the corresponding
Hamiltonian would have infinitely many bound states inside the
spectral gap regardless of the strength of the dipole moment. These
bound states should be supported at long distances and small momenta
where the non-relativistic behavior of the operator -- due to the mass
gap -- is dominant. It is thus plausible that their result agrees with 
the prediction of Connolly and Griffiths \cite{ConnollyGriffiths2007} for the
two-dimensional Schr\"odinger operator.
In contrast, for the three-dimensional Schr\"odinger operator
there is a critical dipole moment below which no bound states exist,
see Abramov and Komarov \cite{AbramovKomarov1972}.  

The argument of De Martino et al is based on replacing the electric
potential by the pure dipole part whose singularity is cut off at
small distances. This approximation is -- physically -- justified,
since -- as pointed out above -- almost all the bound states are
supported at large distances where the dipole approximation is
good. Based on this approximation the problem is explicitly solvable
in terms of Mathieu functions and McDonald functions.

In this paper we will show that the result can indeed be proven and --
in fact -- be generalized -- up to technical constraints -- to
arbitrary charge distributions of total vanishing charge.  Indeed, the
non-vanishing of either the total charge or the dipole moment is
necessary and sufficient for the existence of infinitely many bound
states.

De Martino et al also predicted exponential clustering of those
eigenvalues $E_n$ as they approach edges of the gap $(-m,m)$. We show
-- in the same vein -- that all the moments of the distance to the
nearest gap edge, i.e., $\sum_n (m-|E_n|)^\delta$, converge for all
positive $\delta$.

\pagebreak

We will use the following notation. Let  $x_0\in\R^2\setminus\{0\}$. The two-dimensional Dirac operator
$D$ is initially given on the dense domain
$\dom_0:=C_{0}^{\infty}(\R^2\setminus\{-x_0,x_0\})$ as
\begin{equation}\label{Dirac operator (symmetric)}
\begin{split}
D&=D_0+\gamma V,\\
D_0&=-\I\sigma\cdot\nabla+m\sigma_3\\
V(x)&=|x-x_0|^{-1}-|x+x_0|^{-1},
\end{split}
\end{equation}
where $\sigma=(\sigma_1,\sigma_2)$ and $\sigma_1,\sigma_2,\sigma_3$ are the standard Pauli matrices. We may assume without loss of generality that the coupling constant $\gamma$ (which plays the role of the dipole moment in the present case) is positive; otherwise, we could just replace $x_0$ by $-x_0$.
Note $D$ is symmetric but not essentially self-adjoint. We will find a distinguished self-adjoint
extension with the property that the kinetic energy remains finite. The punctured
plane $\R^2\setminus\{-x_0,x_0\}$ is chosen here because of the Coulomb singularities of the potential could be replaced by $\R^2$ for regular $V$.

We write $\mathcal{B}(\mathcal{H},\mathcal{K})$ for the bounded
operators from a Hilbert space $\H$ to a Hilbert space
$\mathcal{K}$. If $\mathcal{K}=\H$, we just write $\mathcal{B}(\mathcal{H})$. The identity in $\mathcal{K}$ is denoted by$I_{\mathcal{K}}$. In the following, we shall set $\H=L^2(\R^2,\C^2)$ and denote its scalar product (linear in the second argument) and norm by $(\cdot,\cdot)$ and $\|\cdot\|$, respectively. Moreover, we use 
$\mathfrak{S}_p$ for the Schatten ideal of order $p$ in $\H$ and
$R_0(z)=(D_0-z)^{-1}$ for the free resolvent.

\section{Self-adjoint extension}
Since the potential has Coulomb singularities, the extension of the
above symmetric operator is not entirely straightforward. In
particular, in the absence of a Hardy inequality in two dimensions it
is -- contrary to the three dimensional case -- not even possible for small
coupling constant to define the operator as an operator sum by means
of the perturbation theory of Kato and Rellich. Instead, we
resort to a resolvent type equation, in the spirit of Kato
\cite{Kato1966W} and Nenciu \cite{Nenciu1976}.

\begin{theorem}[Existence of a distinguished self-adjoint extension]
  \label{Theorem self-adjoint extension} 
  Assume that $\gamma<1/2$. Then there exists a unique
  self-adjoint extension $D_{\rm ex}$ of $D$ with the property
  $\dom(D_{\rm ex})\subset H^{1/2}(\R^2,\C^2)$.
\end{theorem}

\begin{proof}
  \underline{Step 1:} We claim that for any $a\in\R^2$, $\eta\in\R$
  and $\psi\in C_0^{\infty}(\R^2,\C^2)$,
  \begin{align}\label{Sandwich inequality 2d}
    \||x-a|^{1/2}(D_0-\I\eta)|x-a|^{1/2}\psi\|^2\geq
    \frac{1}{4}\|\psi\|^2
  \end{align}
  By translation invariance of $D_0$ it is sufficient to prove
  \eqref{Sandwich inequality 2d} for $a=0$. We write $D_0$ in polar
  coordinates $(r,\theta)$,
  \begin{align*}
    D_0=\begin{pmatrix}
      m&\e^{-\I\theta}\left(-\I\partial_r-\frac{1}{r}\partial_{\theta}\right)\\
      \e^{\I\theta}\left(-\I\partial_r+\frac{1}{r}\partial_{\theta}\right)&-m
    \end{pmatrix}.
  \end{align*}
  Then, for $\psi\in C_0^{\infty}(\R^2)$,
  \begin{align*} 
  \|r^{1/2}(D_0-\I\eta)r^{1/2}\psi\|^2&=(m^2+\eta^2)\|r\psi\|^2+\|r^{1/2}\partial_r r^{1/2}\psi\|^2+\|\partial_{\theta}\psi\|^2\\
  &\geq \|r^{1/2}\partial_r r^{1/2}\psi\|^2.
  \end{align*}
  Setting $\chi=r\psi$, and integrating by parts, we obtain
   \begin{align*}
   \|r^{1/2}\partial_r r^{1/2}\psi\|^2=\int_0^{2\pi}\int_{0}^{\infty}|\partial_r\chi|^2 r\rd r\rd \theta +\frac{1}{4}\int_0^{2\pi}\int_{0}^{\infty}\frac{|\chi|^2}{r^2}r\rd r\rd \theta\geq\frac{1}{4}\|\psi\|^2,
   \end{align*}
which proves \eqref{Sandwich inequality 2d}. Incidentally, the constant $1/4$ in \eqref{Sandwich inequality 2d} is sharp. This fact becomes apparent in the invariant subspace decomposition of $D_0$ with respect to the total angular momentum $J=-\I\partial_{\theta}+\frac{1}{2}\sigma_3$, and is related to the sharp one-dimensional Hardy inequality.

\underline{Step 2:} We first consider the case of one Coulomb
singularity. We introduce the scale of spaces
\begin{align*}
\mathcal{H}^{+}\subset\mathcal{H}\subset\mathcal{H}^{-},\quad \mathcal{H}^{\pm}:=
H^{\pm 1/2}(\R^2,\C^2),
\end{align*}
where the embeddings are dense and continuous. As is customary, we shall denote the duality pairing 
in $\H^+\times\H^{-}$ by $(\cdot,\cdot)$ as well.
Obviously, 
\begin{align}\label{D0 and inverse on scale space}
D_0\in\mathcal{B}(\mathcal{H}^+,\mathcal{H}^-),\quad R_0(\I\eta)\in\mathcal{B}(\mathcal{H}^-,\mathcal{H}^+).
\end{align}
Following the method of Kato \cite{Kato1983} we show that   
\begin{equation}\label{definition extension}
\begin{split}
  \dom(D_a)&:=\set{\psi\in \mathcal{H}^{+}}{(D_0+\gamma |x-a|^{-1})\psi\in \mathcal{H}},\\
  D_{a}\psi&:=(D_0+\gamma |x-a|^{-1})\psi,
  \end{split}
\end{equation}
is a self-adjoint operator. By the basic criterion for self-adjointness
\cite[Thm. VIII.3]{ReedSimon1972}, it is sufficient to show that
$D_{a}$ is symmetric and that $\Ran(D_{a}\pm \I)=\mathcal{H}$. Since
$\dom_0\subset\dom(D_{a})$, the operator $D_a$ is densely defined. To
prove that $D_a$ is symmetric, it remains to show that
\begin{align}\label{symmetry}
(D_a\phi,\psi)=(\phi,D_a\psi),\quad \phi,\psi\in\dom(D_a).
\end{align}
For later use, we recall the following generalized Hardy inequality \cite{{Herbst1977}}. Let $0<\alpha<n$. Then on $H^{\alpha/2}(\R^n)$, 
\begin{align}\label{Herbst inequality}
|\sqrt{-\Delta}|^{\alpha}-2^a\left[\frac{\Gamma\left(\frac{n+a}{4}\right)}{\Gamma\left(\frac{n-a}{4}\right)}\right]^2|x|^{-\alpha}> 0,
\end{align}
and the inequality continues to hold (with the same sharp constant) if $\sqrt{-\Delta}$ is replaced by $\sqrt{-\Delta+m^2}$ and/or $|x|$ is replaced by $|x-a|$ (by translation invariance).
In particular, \eqref{Herbst inequality} (with $n=2$, $\alpha=1/2$) implies that
\begin{align}\label{bounded operators 1}
|x-a|^{-1/2}\in\mathcal{B}(\mathcal{H}^+,\mathcal{H})\cap \mathcal{B}(\mathcal{H},\mathcal{H}^-),
\end{align}
which, together with \eqref{D0 and inverse on scale space}, implies that
\begin{align}\label{bounded operators 2}
|x-a|^{-1}\in\mathcal{B}(\mathcal{H^+},\mathcal{H}^-),\quad D_0+\gamma|x-a|^{-1}\in\mathcal{B}(\mathcal{H^+},\mathcal{H}^-).
\end{align}
Let $\phi,\psi\in\dom(D_a)\subset\mathcal{H}^+$. By \cite[Thm. 7.14]{LiebLoss1996}, there exist $(\psi_n)_n\subset C^{\infty}_0(\R^2)$ such that $\psi_n\to \psi$ in~$\mathcal{H}^{+}$. 
By the definition of the weak derivative and \eqref{bounded operators 2}, 
\begin{align*}
(D_a\phi,\psi)&=((D_0+\gamma|x-a|^{-1})\phi,\psi)=
\lim_{n\to\infty}\left((D_0\phi,\psi_n)+(\gamma|x-a|^{-1}\phi,\psi_n)\right)\\
&=\lim_{n\to\infty}\left((\phi,D_0\psi_n)+(\phi,\gamma|x-a|^{-1}\psi_n)\right)
=\lim_{n\to\infty}(\phi,(D_0+\gamma|x-a|^{-1})\psi_n)\\
&=(\phi,(D_0+\gamma|x-a|^{-1})\psi)=(\phi,D_a\psi).
\end{align*}
This proves \eqref{symmetry}.

To show that $\Ran(D_a\pm \I)=\H$, observe that by \eqref{D0 and inverse on scale space}, \eqref{bounded operators 1}
\begin{align}\label{Q}
Q(\I\eta):=|x-a|^{-1/2}R_0(\I\eta)|x-a|^{-1/2}\in\mathcal{B}(\mathcal{H}).
\end{align}
Moreover, \eqref{Sandwich inequality 2d} implies $\|Q(\I\eta)\|_{\mathcal{B}(\mathcal{H})}\leq 2.$
By the Neumann series, for $\gamma<1/2$, the operator
\begin{align}\label{resolvent formula}
R(\I\eta):=R_0(\I\eta)-\gamma R_0(\I\eta)|x-a|^{-1/2}\left(I+\gamma Q(\I\eta)\right)^{-1}|x-a|^{-1/2}R_0(\I\eta)
\end{align}
is in $\mathcal{B}(\mathcal{H}^-,\mathcal{H}^+)$.
A straightforward computation shows that 
\begin{equation}\label{left right inverses}\begin{split}
R(\I\eta)(D_0+\gamma|x-a|^{-1}-\I\eta)&=I_{\mathcal{H}^+},\\
(D_0+\gamma|x-a|^{-1}-\I\eta) R(\I\eta))&=I_{\mathcal{H}^-},
\end{split}
\end{equation}
compare \cite{Kato1983}. Let $\psi\in \mathcal{H}\subset\mathcal{H}^-$. Then $\phi=R_0(\I\eta)\psi\in \mathcal{H}^+$, and by the second identity in \eqref{left right inverses}, 
$(D_0+\gamma |x-a|^{-1})\phi=\psi\in\mathcal{H}$,
so that $\phi\in\dom(D_a)$, and $D_a\phi=\psi$. This completes the proof of $\Ran(D_a\pm \I)=\H$.

\underline{Step 3:} Following Nenciu \cite{Nenciu1977}, we extend the
above proof to the two-center potential $V=V_1+V_2$, where
\begin{align*}
V_1(x)=\frac{1}{|x-x_0|},\quad V_2(x)=\frac{1}{|x+x_0|}.
\end{align*}
 Let $\chi\in C_0^{\infty}(\R_+)$ be a nonnegative function such that $\chi(r)=1$ for $r\leq |x_0|/4$ and $\chi(r)=0$ for $r\geq |x_0|/2$, and let 
\begin{align*}
\widetilde{V}_1(x):=\chi^2(|x-x_0|)V_1(x),\quad \widetilde{V}_2(x):=\chi^2(|x+x_0|)V_2(x).
\end{align*}
We split $V$ into a singular and a regular part,
$V=\widetilde{V}+(V-\widetilde{V})$, where
$\widetilde{V}:=\widetilde{V}_1+\widetilde{V}_2$. Note that the
analogues of \eqref{bounded operators 1}--\eqref{bounded operators 2}
hold for $\widetilde{V}$, $\widetilde{V}_i$, $i=1,2$, while
$V-\widetilde{V}\in\mathcal{B}(\mathcal{H})$. We will use
\eqref{Sandwich inequality 2d} to show that for any $\varepsilon>0$
there exists $\eta_0>0$ such that
\begin{align}\label{Sandwich inequality multicenter}
\||\widetilde{V}|^{1/2}R_0(\I\eta)|\widetilde{V}|^{1/2}\|_{\mathcal{B}(\mathcal{H})}\leq \left(2+\varepsilon \right),\quad |\eta|>\eta_0.
\end{align}
Repeating the arguments of the last step, one then sees that the
operator $\widetilde{D}$, defined as in \eqref{definition extension},
but with $|x-a|^{-1}$ replaced by $\widetilde{V}$, is a self-adjoint
operator for $\gamma<1/2$. Self-adjointness of $D_{\rm
  ex}:=\widetilde{D}+\gamma(V-\widetilde{V})$ then follows from the
Kato-Rellich theorem \cite{ReedSimon1972}.  Indeed, upon substituting
$|x-a|^{-1/2}$ in \eqref{Q}--\eqref{resolvent formula} by $|V|^{1/2}$
and $V^{1/2}$ in the first, respectively in the second occurrence, one
checks that $R(\I\eta)\in\mathcal{B}(\mathcal{H}^-,\mathcal{H}^+)$ is
the inverse of
$D_0+\gamma\widetilde{V}-\I\eta\in\mathcal{B}(\mathcal{H}^+,\mathcal{H}^-)$. Here,
$V^{1/2}:=|V|^{1/2}U$ where $U$ is the partial isometry in the polar
decomposition of $V$.  Note that, by the support properties of $\chi$,
we have
\begin{align}\label{support properties chi}
|\widetilde{V}|^{1/2}=|\widetilde{V}_1|^{1/2}+|\widetilde{V}_2|^{1/2},
\end{align}
so that by the triangle inequality, we have for $\psi\in\mathcal{H}$,
\begin{align*}
\||\widetilde{V}|^{1/2}R_0(\I\eta)|\widetilde{V}|^{1/2}\psi\|^2\leq A_1^2+A_2^2+2B(A_1+A_2)+B^2,
\end{align*}
with 
\begin{align*}
A_i:=\||\widetilde{V}_i|^{1/2}R_0(\I\eta)|\widetilde{V}_i|^{1/2}\psi\|^2,\quad B:=\sum_{i\neq j}\||\widetilde{V}_i|^{1/2}R_0(\I\eta)|\widetilde{V}_j|^{1/2}\psi\|^2\quad i,j=1,2.
\end{align*}
By \eqref{Sandwich inequality 2d}, 
\begin{align*}
A_1^2&=\|\chi(|x-x_0|)|V_1|^{1/2}R_0(\I\eta)\chi(|x-x_0|)|V_1|^{1/2}\psi\|^2\\
&\leq \||V_1|^{1/2}R_0(\I\eta)\chi(|x-x_0|)|V_1|^{1/2}\psi\|^2\\
&\leq 4\|\chi(|x-x_0|)\psi\|^2,
\end{align*}
and similarly for $A_2^2$. 
Therefore,
\begin{align*}
A_1^2+A_2^2\leq 4\left(\|\chi(|x-x_0|)\psi\|^2+\|\chi(|x+x_0|)\psi\|^2\right)\leq 
4\|\psi\|^2.
\end{align*}
To finish the proof of \eqref{Sandwich inequality multicenter}, we claim that
\begin{align}\label{ineqij to zero}
\lim_{|\eta|\to\infty}\frac{\||\widetilde{V}_i|^{1/2}R_0(\I\eta)|\widetilde{V}_j|^{1/2}\psi\|^2}{\|\psi\|^2}=0,\quad i\neq j.
\end{align}
This follows from the following estimate for the free resolvent kernel. For $k,l=1,2$, $|x-y|\geq |x_0|$ and $|\eta|\geq \eta_0$,
\begin{align}\label{estimate for the free resolvent kernel}
|R_0(\I\eta)_{kl}(x-y)|\leq C(x_0,\eta_0)\e^{-\frac{1}{4}\sqrt{m^2+\eta^2}|x-y|}.
\end{align}
Indeed, assuming \eqref{estimate for the free resolvent kernel} for the moment, it follows that the Hilbert-Schmidt norm of $|\widetilde{V}_i|^{1/2}R_0(\I\eta)|\widetilde{V}_j|^{1/2}$ is bounded by
\begin{align*}
 4C(x_0,\eta_0)\e^{-\frac{1}{4}\sqrt{m^2+\eta^2}|x_0|}\|\widetilde{V}_i\|_{L^1(\R^2)}\|\widetilde{V}_j\|_{L^1(\R^2)},
\end{align*}
and this converges to zero as $|\eta|\to\infty$. Since the operator
norm is bounded by the Hilbert-Schmidt norm, \eqref{ineqij to zero}
follows. It remains to prove \eqref{estimate for the free resolvent
  kernel}. Noticing that
\begin{align*}
R_0(\I\eta)=(D_0+\I\eta)(-\Delta+k^2)^{-1}, \quad \kappa^2:=m^2+\eta^2,
\end{align*}
and using 
the explicit formula for the heat kernel of $-\Delta$, we arrive at
\begin{align*}
|R_0(\I\eta)_{kl}(x-y)|&=\left|\frac{1}{4\pi}\int_0^{\infty}\left(D_0+\I\eta\right)_{kl}\e^{-\kappa^2 t}\e^{-\frac{|x-y|^2}{4t}}\frac{\rd t}{t}\right|\\
&\leq \frac{1}{4\pi}\int_0^{\infty}\left(\frac{|x-y|}{2t}+\kappa\right)\e^{-\kappa^2 t}\e^{-\frac{|x-y|^2}{4t}}\frac{\rd t}{t}\\
&\leq C(x_0,\eta_0)\e^{-\frac{1}{4}\kappa|x-y|},
\end{align*}
for $|x-y|\geq |x_0|$, $|\eta|\geq \eta_0$ and $k,l=1,2$; the constant can be taken e.g.\ as
\[
C(x_0,\eta_0):=\frac{1}{4\pi}\left(\frac{4}{|x_0|}+\frac{16}{|x_0|^2\sqrt{m^2+\eta_0^2}}\right).
\]

\underline{Step 4:} To prove the uniqueness statement of the Theorem,
suppose that there is another self-adjoint extension $H\supset D$ such
that $\dom(H)\subset \mathcal{H}^+$. Let
$\phi\in\dom(H)\subset\mathcal{H}^+$ and
$\psi\in\dom(D)=\dom_0$. Regarding $D_0+\gamma V$ as an operator in
$\mathcal{B}(\mathcal{H^+},\mathcal{H^-})$ again and repeating the
integration by parts argument in the proof of \eqref{symmetry}, we
obtain
\begin{align*}
(H\phi,\psi)=(\phi,H\psi)=(\phi,D\psi)=(\phi,(D_0+\gamma V)\psi)=((D_0+\gamma V)\phi,\psi).
\end{align*}
Since $\mathcal{D}_0$ is dense in $\mathcal{H}$, this implies $(D_0+\gamma V)\phi=H\phi\in \mathcal{H}$. Hence, $\phi\in\dom(D_{\rm ex})$, and $H\phi=D_{\rm ex}\phi$. This proves that $H\subset D_{\rm ex}$. The reverse inclusion is proved similarly. 
\end{proof}

\begin{remark}
  The proof can easily be extended to cover the case of $N$ Coulomb
  singularities, see \cite{Nenciu1977} for the three-dimensional case.
\end{remark}

\begin{proposition}\label{prop. essential spectrum}
The essential spectrum of $D_{\rm ex}$ is 
\begin{align*}
\sigma_{\rm ess}(D_{\rm ex})=\sigma_{\rm ess}(D_0)=(-\infty,-m]\cup [m,\infty).
\end{align*}
\end{proposition}

\begin{proof}
  We show that the resolvent difference of $D_{\rm ex}$ and $D_{0}$ is
  compact. The claim then follows from Weyl's essential spectrum
  theorem \cite[Thm. XIII.14]{ReedSimon1978}. As in the proof of
  Theorem \ref{Theorem self-adjoint extension} let $\widetilde{D}$ be
  the self-adjoint operator corresponding to the singular part of $V$,
  and denote its resolvent by $\widetilde{R}(\I\eta)$.  By the
  Kato-Seiler-Simon inequality \cite[Thm. 4.1]{Simon1979T},
\begin{align}\label{Kato Seiler Simon}
\||\widetilde{V_i}|^{1/2}R_0(\I\eta)\|_{\mathfrak{S}_p}\leq C \||\widetilde{V_i}|^{1/2}\|_{p}\|(|\cdot|^2+m^2)^{-1/2}\|_{p},
\end{align}
and the right hand side is finite for all $p\in(2,4)$.
By \eqref{support properties chi} and the triangle inequality, \eqref{Kato Seiler Simon} continues to hold (with $2C$) if $\widetilde{V_i}$ is replaced by $\widetilde{V}$. 
The analogue of the resolvent formula \eqref{resolvent formula} for $\widetilde{D}$ and the trace ideal property of $\mathfrak{S}_p$ then imply that $\widetilde{R}(\I\eta)-R_0(\I\eta)\in\mathfrak{S}_p$ for all $p>1$,
in particular it is compact.
Denoting by $R(\I\eta)$ the resolvent of $D_{\rm ex}=\widetilde{D}+\gamma (V-\widetilde{V})$, we have 
\begin{align*}
R(\I\eta)-R_0(\I\eta)&=-\gamma R(\I\eta)(V-\widetilde{V})\widetilde{R}(\I\eta)+(\widetilde{R}(\I\eta)-R_0(\I\eta)).
\end{align*}
It remains to be shown that first summand is compact. Indeed, its $\mathfrak{S}_p$-norm is bounded by
\begin{align*}
\gamma \|R(\I\eta)\| \|(V-\widetilde{V})(I-\Delta)^{-1/4}\|_{\mathfrak{S}_p}\|(I-\Delta)^{1/4}\widetilde{R}(\I\eta)(I-\Delta)^{1/4}\|\|(I-\Delta)^{-1/4}\|,
\end{align*}
which is finite for $p>4$ by \cite[Thm. 4.1]{Simon1979T}.
\end{proof}

\section{Existence of infinitely many eigenvalues}

\begin{theorem}\label{thm. infinitely many eigenvalues dipole}
Any self-adjoint extension of $D$ (defined in \eqref{Dirac operator (symmetric)}) has infinitely many eigenvalues in $(-m,m)$.
\end{theorem}

\begin{proof}
  Let $H$ be a self-adjoint extension of $D$. Then $H^2$, defined by
  the spectral theorem, is the unique operator associated to the
  nonnegative symmetric form
\begin{align*}
q(\phi,\psi):=(H\phi,H\psi),\quad \psi\in\dom(q):=\dom(H)
\end{align*}
by the first representation theorem \cite[Thm. 2.1]{Kato1966}. Indeed, the form
$q$ is closed since $H$ is (self-adjoint and hence) closed. Let $T$ be
the self-adjoint operator associated to the form $q$ by the first
representation theorem. Since $(H\phi,H\psi)=(T\phi,\psi)$ for all $\phi\in\dom(T)$ and $\psi\in\dom(H)$, it follows that 
$T\subset H^2$. Since $T$ is self-adjoint, we have $T=H^2$.

Let $q_0$ be the nonnegative symmetric form
\begin{align*}
q_0(\phi,\psi):=(D\phi,D\psi),\quad \psi\in\dom(q):=\dom(D)=\dom_0.
\end{align*}
We use the Cauchy-Schwarz inequality to obtain
\begin{align*}
q_0[\psi]&=\|\nabla\psi\|^2+\gamma^2\|V\psi\|^2+2\gamma\re\left(-\I\sigma\cdot\nabla\psi,V\psi\right)+2m\gamma(\sigma_3V\psi,\psi)\\
&\leq 
2\|\nabla\psi\|^2+2\gamma^2\|V\psi\|^2+2m\gamma(\sigma_3V\psi,\psi)=:s_+[\psi_+] +s_-[\psi_-],
\end{align*}
with $\psi=(\psi_+,\psi_-)^T$ and
\begin{align*}
s_{\pm}[\psi_{\pm}]:=\|\nabla \psi_{\pm}\|^2+\gamma^2\|V\psi_{\pm}\|^2\pm \gamma m (V\psi_{\pm},\psi_{\pm}),\quad \dom(s_{\pm})=\dom_0.
\end{align*}
Clearly, $q_0\subset q$, which (by the variational principle) implies that
\begin{align*}
N(H\in (-m,m))&=N(H^2-m<0)=\sup_{M\subset\dom(q)}\set{\dim M}{q[\psi]<0,\,\psi\in M}\\
&\geq  \sup_{M\subset\dom(q_0)}\set{\dim M}{q_0[\psi]<0,\,\psi\in M}\\
&=\sup_{M\subset\dom(s_+)}\set{\dim M}{s_+[\psi]<0,\,\psi\in M}\\
&+ \sup_{M\subset\dom(s_-)}\set{\dim M}{s_-[\psi]<0,\,\psi\in M}. 
\end{align*}
It is thus sufficient to show that there exist infinitely many
orthonormal functions $\varphi_n\in\dom_0$ such that
$s_{-}[\varphi_n]<0$. Note that we could as well have chosen $s_+$
because of the symmetry $s_+[U\psi]=s_-[\psi]$, where
$U\psi(x):=\psi(x-2 x\cdot x_0/|x_0|)$ is a unitary transformation.

Without loss of generality, we may assume that $x_0=e_1$.
In polar coordinates (by Taylor's theorem) we then have
\begin{align*}
V(r,\theta)=-2\frac{\cos\theta}{r^2}+O(r^{-3}).
\end{align*}
For $k>1$ define the radially symmetric function 
\begin{align}\label{chi}
\chi(r):=\begin{cases}
0\quad &r\leq k,\\
\frac{r-k}{k^2-k}\quad &k \leq r\leq k^2,\\
1\quad &k^2 \leq r\leq k^3,\\
\frac{k^4-r}{k^4-k^3}\quad &k^3 \leq r\leq k^4,\\
0\quad &k^4\leq r.
\end{cases}
\end{align}
We set $\chi_R(r):=R^{-1}\chi(r/R)$. Moreover, let $Y_0(q;\cdot)$ be the normalized eigenfunction corresponding to the lowest eigenvalue $\lambda_0(q)$ of the Mathieu operator
\begin{align}\label{Mathieu operator}
M(q)=-\partial_{\theta}^2+2q\cos\theta
\end{align}
on $L^2(S^1)$. It is known that $\lambda_0(q)<0$ for any $q>0$, see Section 2.150, Formula~(7) in \cite{McLachlan1947}. Setting
\begin{align}\label{trial eigenfunction}
\psi_R(r,\theta):=\chi_R(r)Y_0(m\gamma;\theta),
\end{align}
we obtain
\begin{align*}
s_-[\psi_R]&=R^{-2}\|\partial_r\chi\|_{L^2(\R_+,r\rd r)}^2+R^{-2}\lambda_0(m\gamma)\|r^{-1}\chi\|_{L^2(\R_+,r\rd r)}^2
+O(k^{-1})\\
&\leq R^{-2}\left(\frac{k^2+k}{k^2-k}+\frac{k^4+k^3}{k^4-k^3}\right)+R^{-2}\lambda_0(m\gamma)\ln k+O(k^{-1}),
\end{align*}
and this is negative for sufficiently large $k$.
Hence, the functions $\varphi_{n}:=\psi_{2^n}/\|\psi_{2^n}\|$ with $2^n>k^3$, are orthonormal and satisfy $s_{\pm}[\varphi_n]<0$ for all such $n$. 
\end{proof}

\begin{remark}
The existence of infinitely many eigenvalues for arbitrarily small dipole moment $\gamma$ is a consequence of the fact that the Mathieu operator \eqref{Mathieu operator} always has a negative eigenvalue for any $q>0$. Moreover, as the dipole moment (and hence $q=m\gamma$) increases, additional negative eigenvalues may emerge. Each time such a threshold is crossed, another infinite sequence of trial functions (with $Y_0$ in \eqref{trial eigenfunction} replaced by any eigenfunction of the Mathieu operator corresponding to a negative eigenvalue) can be constructed. These infinite sequences, labeled by the negative eigenvalues of the Mathieu operator, were called ''towers`` in \cite{Martinoetal2014}.
\end{remark}

\section{Clustering of eigenvalues at the edges of the gap}

In the following theorem, we denote by $C_H$ the constant in \eqref{Herbst inequality} for $n=2$, $a=1$,
\begin{align*}
C_H:=\frac{4\pi^2}{\Gamma(1/4)^4}\approx 0.229.
\end{align*}

\begin{theorem}
  Let $\delta>0$ and $\gamma<C_H$. Then the eigenvalues $E_n$ of $D_{\rm ex}$ satisfy
  \begin{equation}
    \label{eigenvalue sums dipole}
    \sum_{n}(m-|E_n|)^{\delta}\leq \frac{L m^{1+\delta-\delta_0}\gamma^{1+\delta_0}|x_0|^{1-\delta_0}}{(1-\gamma/C_H)^{2+\delta_0}}\frac{1}{\delta_0(1-\delta_0)}
  \end{equation}
for any $\delta_0\in(0,1)$ such that $\delta_0\leq \delta$; here, $L$ is some universal constant.
\end{theorem}

\begin{proof}
We follow the lines of the proof of Frank and Simon for the one-dimensional Dirac operator \cite[Thm. 7.1]{FrankSimon2011}. The main tool in their proof, Theorem 1.4 in \cite{FrankSimon2011}, is stated for relatively compact perturbations, but still applies if the resolvent difference of the perturbed and unperturbed operator is compact; this is the case here, by Proposition \ref{prop. essential spectrum}. Proceeding as in \cite[Thm. 7.2]{FrankSimon2011}, one can then show that
\begin{align*}
\sum_{n}(m-|E_n|)^{\delta}\leq 2\left[\tr(H_0-\gamma V_-)_-^{\delta}+\tr(H_0-\gamma V_+)_-^{\delta}\right],
\end{align*}
where $H_0:=\sqrt{|p|^2+m^2}-m$
and $V_{\pm}$ are the positive and negative parts of $V$, respectively.
By decomposing $H_0$ into a part with small momentum and a part with large momentum, one can estimate
\begin{align}\label{sum of nonrelativistic and relativistic}
\tr(H_0-\gamma V_{\pm})_-^{\delta}\leq \tr\left(\frac{c_1|p|^2}{m}-\theta^{-1}\gamma V_{\pm}\right)_-^{\delta}+\tr\left(c_2|p|-(1-\theta)^{-1}\gamma V_{\pm}\right)_-^{\delta}
\end{align}
where
$c_1=(\sqrt{\rho^2+1}-1)\rho^{-2}$, $c_2=(\sqrt{\rho^2+1}-1)\rho^{-1}$,
and where $\rho>0$ and $0<\theta<1$ are arbitrary parameters, see
\cite[(7.9)--(7.12)]{FrankSimon2011}. Since $V_{\pm}$ decay like
$|x|^{-2}$ at infinity,
\begin{align}\label{first part of eigenvalue sum finite}
\tr\left(\frac{c_1|p|^2}{m}-\theta^{-1}\gamma V_{\pm}\right)_-^{\delta}\leq c_1^{-1}\theta^{-1-\delta}m L^{\rm LT}_{\delta,2}\int_{\R^2}(\gamma V_{\pm})^{1+\delta}\rd x<\infty
\end{align}
for all $\delta\in (0,1)$, where $L^{\rm LT}_{\delta,2}$ is the best constant in the Lieb-Thirring inequality. The case $\delta\geq 1$ is prohibited by
the singularities of $V_{\pm}$ at $\pm x_0$; however, the left hand
side of \eqref{eigenvalue sums dipole} is clearly finite for all
$\delta\geq \delta_0$ if it is finite for $\delta_0$ since
\begin{align}\label{sum delta geq 1}
  \sum_{n}(m-|E_n|)^{\delta}
  \leq m^{\delta-\delta_0}\sum_{n}\left(m-|E_n|\right)^{\delta_0}.
\end{align}

We now show that the second term in \eqref{sum of nonrelativistic and relativistic} is in fact zero.
We may assume that~$x_0=|x_0|e_1$. Then
\begin{align*}
V_+(x)&=V(x)\chi\{x_1\geq 0\}\leq |x-x_0|^{-1},\\
V_-(x)&=-V(x)\chi\{x_1\leq 0\}\leq |x+x_0|^{-1}.
\end{align*}
Hence, by Hardy's generalized inequality \eqref{Herbst inequality},
\begin{align*}
c_2|p|-(1-\theta)^{-1}\gamma V_{\pm}\geq c_2|p|-(1-\theta)^{-1}\gamma|x\mp x_0|^{-1}>0,
\end{align*}
provided $\gamma\leq c_2(1-\theta)C_H$. We will choose $\theta$
such that equality holds. Moreover, we pick $\rho$ such that $c_2=(1+\gamma/C_H)/2$
and evaluate the bound \eqref{first part of eigenvalue sum finite}. For $\delta\in (0,1)$, we estimate the integral in \eqref{first part of eigenvalue sum finite} in the regions $|x|\leq 2|x_0|$ and $|x|\geq 2|x_0|$, for $\delta\geq 1$, we use \eqref{sum delta geq 1}.
\end{proof}

\section{General charge distributions}

Let $\mu$ be a signed Borel measure 
on $\R^3$.
The corresponding potential is
\begin{align}\label{potential general charge distribution}
V(x)=\frac{1}{4\pi}\int_{\R^3}\frac{\rd\mu(y)}{|x-y|}.
\end{align} 
The physically relevant potential is the restriction of $V$ to
the hyperplane $x_3=0$.

If we assume that $\mu$ has compact support, $\supp(\mu)\subset B(0,R)$, then 
the multipole expansion of $V$ is given by
\begin{align}\label{Multipole expansion}
V(x)=\sum_{l=0}^{\infty}\sum_{m=-l}^l\frac{1}{2l+1}q_{lm}\frac{Y_{lm}(x/|x|)}{|x|^{l+1}},\quad |x|\geq 2R,
\end{align}
with the multipole moments
\begin{align*}
q_{lm}=\int_{\R^3}Y_{lm}(y/|y|)|y|^l\rd \mu(y).
\end{align*}
Note that \eqref{Multipole expansion} converges absolutely and uniformly.
Denote
\begin{align*}
e&=q_{00}=\int_{\R^3}\rd \mu(y),\quad \mbox{(total charge)},\\
p_i&=q_{1i}=\int_{\R^3}y_i\rd \mu(y),\quad i=-1,0,1,\quad \mbox{(dipole moment)}
\end{align*}
and $p=(p_{-1},p_0,p_1)$.
In the next theorem, we show that the condition $e=p=0$ is
necessary and sufficient for the finiteness of the number of
eigenvalues (at least for absolutely continuous measures).
\begin{theorem}\label{thm. infinitely or finitely many general charge distr.}
  Let $\mu$ be absolutely continuous with respect to
  (three-dimensional) Lebesgue measure, with compactly supported
  density $\rho$. Then the number of eigenvalues of $D_0+V$ in
  $(-m,m)$ is finite if and only if $e=p=0$.
\end{theorem}

\begin{remark}
Under the assumptions on the density $\rho$, the potential \eqref{potential general charge distribution}
is a bounded function, and hence $D_0+V$ is self-adjoint on $\dom(D_0)=H^1(\R^2,\C^2)$ by the Kato-Rellich theorem.
\end{remark}

\begin{proof}[Proof of theorem \ref{thm. infinitely or finitely many general charge distr.}]
  If $e\neq 0$ or $p\neq 0$, a straightforward adaptation of the proof
  of Theorem \ref{thm. infinitely many eigenvalues dipole}, using the
  multipole expansion \eqref{Multipole expansion}, shows that the
  there are infinitely many eigenvalues in $(-m,m)$. In the former
  case, we just replace the test functions $\psi_R$ by the radial
  functions $\chi_R$.

Let $e=p=0$, and let $l\geq 2$ be the least integer for which not all $q_{lm}$ are zero. Then \eqref{Multipole expansion} and the boundedness of $V$ imply that $|V(r\e^{\I\phi},0)|\leq  C_lq_lW_l(r)$, where $W_l(r):=(1+r)^{-l-1}$, $q_l=\max_{-l\leq m\leq l}|q_{lm}|$, and $C_l>0$ is a constant.
Hence,
\begin{align*}
\|(D_0+V)\psi\|^2&\geq \frac{1}{2}\|\nabla\psi\|^2-\|V\psi\|^2-m(|V|\psi,\psi)\\
&\geq \frac{1}{2}\|\nabla\psi\|^2-C_l^2q_l^2\|W_l\psi\|^2-m C_lq_l(W_l\psi,\psi)
\end{align*} 
and
\begin{equation}\label{number of eigenvalues for e and p zero}
\begin{split}
N(D_0+V\in (-m,m))&=N((D_0+V)^2-m<0)\\
&\leq N(-\Delta-C_l^2q_l^2 W_l^2-mC_lq_l W_l<0). 
\end{split}
\end{equation}
Since 
\begin{align*}
\int_0^{\infty}r(W_l(r)+W_l(r)^2)\rd r<\infty,
\end{align*}
the Bargmann-type bounds in \cite{Newton1983} imply that the rightmost quantity in \eqref{number of eigenvalues for e and p zero} is bounded by $1+C_l' (mq_l+q_l^2)$ for some constant $C_l'$. Note that an upper bound to the right hand side in inequality (2) in \cite{Newton1983} is easily obtained by replacing the logarithm by a small power.
\end{proof}

We have seen that the moments $\sum_j(m-|E_j|)^{\delta}$ for the pure dipole potential $V$ in \eqref{Dirac operator (symmetric)} are finite for all $\delta>0$, while for $e=p=0$ they are finite for all $\delta\geq 0$. Under rather general assumptions on the density (in particular, the monopole moment $e$ is not assumed to be zero), the following theorem asserts that the moments exist at least for $\delta>1$.

\begin{theorem}
Let $\delta>1$ and $\rho\in L^{\frac{3(2+\delta)}{2(3+\delta)}}(\R^3)\cap L^{\frac{3(2+\delta)}{2(3+\delta)}}(\R^3)$.
Then, the eigenvalues $E_n$ of $D_0+V$ satisfy
\begin{align*}
\sum_{n}(m-|E_n)^{\delta}\leq C_{\delta}\left(m\|\rho\|_{L^{\frac{3(1+\delta)}{2(2+\delta)}}(\R^3)}^{1+\delta}+\|\rho\|_{L^{\frac{3(2+\delta)}{2(3+\delta)}}(\R^3)}^{2+\delta}\right).
\end{align*}
\end{theorem}

\begin{proof}
The claim follows from \eqref{sum of nonrelativistic and relativistic} and the (relativistic and non-relativistic) Lieb-Thirring inequalities, upon estimating the corresponding Lebesgue norms of $V$ in terms of $\rho$ by means of
the sharp trace inequality in \cite[Thm.~2]{Adams1971}. Note also that in view of Sobolev embedding, $V$ is relatively bounded with respect to $D_0$, with relative bound zero; in particular, $D_0+V$ is self-adjoint.
\end{proof}

\textsc{Acknowledgment:} \textit{We thank Reinhold Egger for
  drawing our attention to the problem and for making \cite{Martinoetal2014}
  available to us before publication. Furthermore, we thank the DFG
  who partially supported this work through the SFB-TR 12.}

\def\cprime{$'$}

\end{document}